\documentclass[a4paper]{article}%
\usepackage{amsmath}
\usepackage{amsfonts}
\usepackage{amssymb}
\usepackage{graphicx}
\usepackage[a4paper]{geometry}
\usepackage{appendix}%
\newtheorem{theorem}{Theorem}[section]

\newtheorem{corollary}[theorem]{Corollary}

\newtheorem{definition}[theorem]{Definition}

\newtheorem{lemma}[theorem]{Lemma}

\newenvironment{proof}[1][Proof]{\noindent\textbf{#1.} }{\ \rule{0.5em}{0.5em}}

\begin{document}

\title{On the optimization of the principal eigenvalue for single-centre
point-interaction operators in a bounded region}
\author{Pavel Exner\thanks{Department of Theoretical Physics, NPI, Academy of
Sciences, 25068 \v{R}e\v{z} near Prague, and Doppler Institute, Czech
Technical University, B\v{r}ehov\'{a} 7, 11519 Prague, Czechia.}, Andrea
Mantile\thanks{Doppler Institute, Czech Technical University, B\v{r}ehov\'{a}
7, 11519 Prague, Czechia, and IRMAR, Universit\'{e} Rennes 1, Campus de
Beaulieu, 35042, Rennes Cedex, France.}}
\date{}
\maketitle

\begin{abstract}
We investigate relations between spectral properties of a single-centre
point-interaction Hamiltonian describing a particle confined to a bounded
domain $\Omega\subset\mathbb{R}^{d},\: d=2,3$, with Dirichlet boundary, and
the geometry of $\Omega$. For this class of operators Krein's formula yields
an explicit representation of the resolvent in terms of the integral kernel of
the unperturbed one, $\left(  -\Delta_{\Omega}^{D}+z\right)  ^{-1}$. We use a
moving plane analysis to characterize the behaviour of the ground-state energy
of the Hamiltonian with respect to the point-interaction position and the
shape of $\Omega$, in particular, we establish some conditions showing how to
place the interaction to optimize the principal eigenvalue.

\end{abstract}

\setcounter{equation}{0}
%%%%%%%%%%%%%%%%%%%%%%%%%%%%%%%%%%%%%%%%%%%%%%%%%%

\section{Introduction}

Relations between geometry of a domain and spectral properties of
corresponding operators belong to the most traditional question in
mathematical physics; one can recall, e.g., the Faber-Krahn inequality
\cite{Fa23, Kr25} or the Payne-P\'olya-Weinberger conjecture \cite{PPW55}
proved by Ashbaugh and Benguria \cite{AB92a, AB92b}. A more recent example
concerns the situation where the domain in question is not simply connected
and one asks, in particular, how to place a circular hard-wall obstacle within
a circular planar cavity to minimize the ground-state eigenvalue; it appears
that the minimum is reached when the obstacle touches the boundary
\cite{Harrell}.

In connection with the last mentioned problem one can also ask what happens if
such a \textquotedblleft hard\textquotedblright\ obstacle is replaced by
another object, say, by a potential barrier or well. In this paper we are
going to address this question in the particular case when such a potential is
singular, in other words, a point interaction. Recalling basic results about
these interactions \cite{Albeverio} we see that the problem makes sense in
dimension $d\leq3$. Furthermore, a simple perturbative argument shows that in
the one dimensional situation the answer may depend on the sign of the
$\delta$ potential; we restrict here our attention to the more singular case
of dimension $d=2,3$.

Problems of this type were to our knowledge solved so far only in cases where
the domain has a simple geometry such as a straight strip in $\mathbb{R}^{2}$,
see \cite{EGST96}, or a planar layer in $\mathbb{R}^{3}$, see \cite{EN02},
where the eigenvalue problem can be solved more or less explicitly. Here we
consider that the domain $\Omega$, to which the particle is confined, is
bounded and otherwise quite general, see below. The operator of interest will
be the corresponding Dirichlet Laplacian perturbed by a single point
interaction of a fixed coupling constant; we will ask about the dependence of
its principal eigenvalue on the perturbation position. Using a figurative
expression, to be made precise below, we are going to show that the
ground-state energy \emph{increases} as the point moves towards the boundary
of $\Omega$, in contrast to the case of a Dirichlet obstacle mentioned above.

Our method is based on the fact that one is able to express the resolvent of
the operator in question by means of the Krein's formula. The resulting
spectral condition allow us to characterize the principal eigenvalue as a
function of the interaction position in a non perturbative setting. Then we
exploit the maximum principle and a domain reflection technique, analogous to
the one used in \cite{Harrell}, to demonstrate our main result,
Theorem~\ref{th: main}, expressing strict monotonicity of the principal
eigenvalue with respect to certain directions. Using this conclusion, we
formulate then some conditions under which the eigenvalue reaches its minimum
value for a fixed $\Omega$.

\setcounter{equation}{0}
%%%%%%%%%%%%%%%%%%%%%%%%%%%%%%%%%%%%%%%%%%%%%%%%%%

\section{Confined point interactions in two and three dimensions}

The definition domain and the spectral properties of point interaction
Hamiltonians in dimensions two and three are usually expressed in terms of the
free Green function, i.e. the integral kernel of the operator $\left(
-\Delta+z\right)  ^{-1}$, cf. \cite{Albeverio}. This is true both if the
configuration space is the whole $\mathbb{R}^{d}$ or if the particle is
confined to a subset $\Omega$ of it by a hard wall corresponding to Dirichlet
boundary condition. In the latter case Green's function is defined by the
equation
%-------------- %
\begin{equation}
\left\{
\begin{array}
[c]{l}%
\medskip\left(  -\Delta+z\right)  \mathcal{G}_{0}^{z}(\underline{x}%
,\underline{x}^{\prime})=\delta(\underline{x}-\underline{x}^{\prime})\\
\left.  \mathcal{G}_{0}^{z}(\underline{x},\underline{x}^{\prime})\right\vert
_{\underline{x}\in\partial\Omega}=0
\end{array}
\right.  ;\quad\underline{x}^{\prime}\in\Omega\label{Green 1}%
\end{equation}
%-------------- %
which admits a solution in $L^{2}(\Omega)$ whenever $-z$ does not
belong to the spectrum of the Dirichlet Laplacian
$-\Delta_{\Omega}^{D}$ defined in the standard way through the
associated quadratic form \cite{Reed4}. Throughout
this paper, $\Omega$ is supposed to be an open regular set in $\mathbb{R}^{d}%
$, $d=2,3$, bounded and connected, not necessarily simply; we assume that
$\partial\Omega$ is piecewise $C^{1}$.

Under these hypotheses $-\Delta_{\Omega}^{D}$ has a purely discrete spectrum;
we denote by $\left\{  \lambda_{n}\right\}  _{n\in\mathbb{N}_{0}}$, where
$\mathbb{N}_{0}:=\mathbb{N\cup}\left\{  0\right\}  $, its eigenvalues, and by
$\left\{  \psi_{n,k}\right\}  _{\substack{k=1\\n\in\mathbb{N}_{0}}}^{N_{n}}$
the corresponding system of eigenfunctions, $N_{n}$ being the multiplicity of
the $n$-th eigenvalue. Projecting (\ref{Green 1}) on the vectors $\psi_{n,k}$,
it is easy to check the validity of the following standard Fourier expansion
of $\mathcal{G}_{0}^{z}$,
%-------------- %
\begin{equation}
\mathcal{G}_{0}^{z}(\underline{x},\underline{x}^{\prime})=\sum_{\substack{n\in
\mathbb{N}_{0}\\k\leq N_{n}}}\frac{\psi_{n,k}(\underline{x}^{\prime}%
)\,\psi_{n,k}(\underline{x})}{\lambda_{n}+z}\,. \label{Green 2}%
\end{equation}
%-------------- %
An alternative representation of the kernel, which will be extensively used in
this paper, can be given in terms of the 'free' Green's function $\mathcal{G}
^{z} (\underline{x}, \underline{x}^{\prime})$, in other words, the integral
kernel of $\left(  -\Delta+z\right)  ^{-1}$ in the whole space\footnote{Here
and in the following we use the convention in which the negative real axis
represents the cut of the square root $\sqrt{z}$ in the complex plane.},%
%-------------- %
\begin{align}
\mathcal{G}^{z}(\underline{x},\underline{x}^{\prime})  &  \!=\!\frac{1}{2\pi
}K_{0}(\sqrt{z}\left\vert \underline{x}-\underline{x}^{\prime}\right\vert
)\quad\mathrm{in}\ \mathbb{R}^{2}\label{Green free_2d}\\
\mathcal{G}^{z}(\underline{x},\underline{x}^{\prime})  &  \!=\!\frac
{e^{-\sqrt{z}\left\vert \underline{x}-\underline{x}^{\prime}\right\vert }%
}{4\pi\left\vert \underline{x}-\underline{x}^{\prime}\right\vert }%
\qquad\mathrm{in}\ \mathbb{R}^{3} \label{Green free_3d}%
\end{align}
%-------------- %
Here $K_{0}$ denotes the Macdonald (or modified Hankel) function.
Using the boundary conditions in (\ref{Green 1}), we can express
the kernel of interest as
%-------------- %
\begin{equation}
\mathcal{G}_{0}^{z}(\underline{x},\underline{x}^{\prime})=\mathcal{G}%
^{z}(\underline{x},\underline{x}^{\prime})-h(\underline{x},\underline
{x}^{\prime},\sqrt{z}) \label{Green 3}%
\end{equation}
%-------------- %
where $\mathcal{G}^{z}$ is defined by (\ref{Green free_2d}%
)--(\ref{Green free_3d}) and $h(\underline{x}, \underline{x}^{\prime},\sqrt
{z})$ solves the boundary value problem%
%-------------- %
\begin{equation}
\left\{
\begin{array}
[c]{l}%
\medskip\left(  -\Delta+z\right)  h(\underline{x},\underline{x}^{\prime}%
,\sqrt{z})=0\\
\left.  h(\underline{x},\underline{x}^{\prime},\sqrt{z})\right\vert
_{\underline{x}\in\partial\Omega}=\left.  \mathcal{G}^{z}(\underline
{x},\underline{x}^{\prime})\right\vert _{\underline{x}\in\partial\Omega}%
\end{array}
\right.  \mathrm{for\;any}\;\;\underline{x}^{\prime}\in\Omega\label{h 1}%
\end{equation}
%-------------- %

In the next step we construct the operator which will be our main
object perturbing $-\Delta_{\Omega}^{D}$ by a single point
interaction with the support at a point
$\underline{x}_{0}\in\Omega$. Such Hamiltonians are defined by the
self-adjoint extensions of the symmetric operator
%-------------- %
\begin{equation}
\left\{
\begin{array}
[c]{l}%
\medskip D(H_{0})=\left\{  \left.  \psi\in H^{2}\cap H_{0}^{1}(\Omega
)\right\vert \,\psi(\underline{x}_{0})=0\right\} \\
H_{0}\psi=-\Delta\psi
\end{array}
\right.  \label{H_0}%
\end{equation}
%-------------- %
Following the von Neumann theory \cite{Akhiezer, Reed2}, we
observe that the restriction (\ref{H_0}) has deficiency indices
$(1,1)$; consequently, we arrive at a one-parameter family of
self-adjoint operators $H_{\alpha}$. For fixed
$\alpha\in\mathbb{R}$, $\lambda\in\mathbb{C}\backslash\mathbb{R}$
and a domain $\Omega\subset\mathbb{R}^{2}$ we have
%-------------- %
\begin{gather}
D(H_{\alpha})=\bigg\{ \psi\in L^{2}(\Omega) \Big\vert\: \psi=\phi^{\lambda}
+q\mathcal{G}_{0}^{\lambda}(\cdot,\underline{x}_{0}),\ \phi^{\lambda}\in
H^{2}\cap H_{0}^{1}(\Omega),\nonumber\\
\qquad\qquad\phi^{\lambda}(\underline{x}_{0})=\frac{q}{2\pi}\,\bigg( \alpha
-\ln\sqrt{\lambda}-2\pi\,h(\underline{x}_{0},\underline{x}_{0} ,\sqrt{\lambda
})\bigg) \,\bigg\} \label{H_alpha_2d dom}%
\end{gather}
%-------------- %
while for $\Omega\subset\mathbb{R}^{3}$ the operator domain is
%-------------- %
\begin{align}
D(H_{\alpha})  &  =\bigg\{\psi\in L^{2}(\Omega)\Big\vert:\psi=\phi^{\lambda
}+q\mathcal{G}_{0}^{\lambda}(\cdot,\underline{x}_{0}),\ \phi^{\lambda}\in
H^{2}\cap H_{0}^{1}(\Omega),\nonumber\\
&  \qquad\phi^{\lambda}(\underline{x}_{0})=q\,\bigg(\alpha+\frac{\sqrt
{\lambda}}{4\pi}+h(\underline{x}_{0},\underline{x}_{0},\sqrt{\lambda
})\bigg)\,\bigg\} \label{H_alpha_3d dom}%
\end{align}
%-------------- %
In both cases the parameter $\lambda$ determines a representation
of the operator domain, roughly speaking, a split between the
regular and singular part; for a fixed choice of $\lambda$ the
action of $H_{\alpha}$ is the following
%-------------- %
\begin{equation}
H_{\alpha}\psi=-\Delta\phi^{\lambda}-\lambda q\,\mathcal{G}_{0}^{\lambda
}(\cdot,\underline{x}_{0}) \,. \label{H_alpha}%
\end{equation}
%-------------- %
It is convenient to include into this description also infinite values of the
parameter $\alpha$, in which case the coefficient $q$ of the singular part ---
sometimes referred to as the \emph{charge} of the state --- vanishes. It is
equivalent to the absence of the point interaction: the domain reduces in this
case to $H^{2}\cap H_{0}^{1}(\Omega)$ and the self-adjoint extensions
corresponding to $\alpha=\pm\infty$ is identified with the unperturbed
operator $-\Delta_{\Omega}^{D}$.

Making use of Krein's formula \cite{Akhiezer} the action of the resolvent,
$R_{z}^{\alpha}=(H_{\alpha}+z)^{-1}$ on $L^{2} (\Omega)$, can be expressed as
a rank-one perturbation of its 'free' counterpart $(-\Delta+z)^{-1}$,
specifically
%-------------- %
\begin{equation}
R_{z}^{\alpha}\varphi=(-\Delta+z)^{-1}\varphi+q_{z}\left(  (-\Delta
+z)^{-1}\varphi\right)  (\underline{x}_{0})\,\mathcal{G}_{0}^{z}
(\cdot,\underline{x}_{0})\,. \label{Krein}%
\end{equation}
%-------------- %
In this formula, of course, the first term is the regular part of
the function $R_{z}^{\alpha}\varphi\in D(H_{\alpha})$ while the
value $q_{z}\left( (-\Delta+z)^{-1}\varphi\right)
(\underline{x}_{0})$ denotes the corresponding charge. Using the
boundary conditions in (\ref{H_alpha_2d dom}) and
(\ref{H_alpha_3d dom}) we can identify the coefficient $q_{z}$ with%
%-------------- %
\begin{equation}
q_{z}=\left(  \alpha-\ln\sqrt{z}-2\pi\,h(\underline{x}_{0},\underline{x}%
_{0},\sqrt{z})\right)  ^{-1},\quad\Omega\subset\mathbb{R}^{2} \label{Krein_2d}%
\end{equation}
%-------------- %
\begin{equation}
q_{z}=\left(  \alpha+\frac{\sqrt{z}}{4\pi}+h(\underline{x}_{0},\underline
{x}_{0},\sqrt{z})\right)  ^{-1},\quad\Omega\subset\mathbb{R}^{3}
\label{Krein_3d}%
\end{equation}
%-------------- %
Furthermore, the Fourier expansion
%-------------- %
\begin{equation}
(-\Delta+z)^{-1}\varphi=\sum_{\substack{n\in\mathbb{N}_{0}\\k\leq N_{n}}%
}\frac{\left(  \varphi,\psi_{n,k}\right)  }{\lambda_{n}+z}\psi_{n,k}
\label{Free propagator}%
\end{equation}
%-------------- %
yields the following explicit expression%
%-------------- %
\begin{equation}
R_{z}^{\alpha}\varphi=\sum_{\substack{n\in\mathbb{N}_{0}\\k\leq N_{n}}%
}\frac{\left(  \varphi,\psi_{n,k}\right)  }{\lambda_{n}+z}\psi_{n,k}+2\pi
\sum_{\substack{n\in\mathbb{N}_{0}\\k\leq N_{n}}}\frac{\left(  \varphi
,\psi_{n,k}\right)  }{\lambda_{n}+z}\frac{\psi_{n,k}(\underline{x}_{0}%
)}{\alpha-\ln\sqrt{z}-2\pi\,h(\underline{x}_{0},\underline{x}_{0},\sqrt{z}%
)}\mathcal{G}_{0}^{z}(\cdot,\underline{x}_{0}) \label{R_alpha_2d}%
\end{equation}
%-------------- %
for $\Omega\subset\mathbb{R}^{2}$, and its counterpart
%-------------- %
\begin{equation}
R_{z}^{\alpha}\varphi=\sum_{\substack{n\in\mathbb{N}_{0}\\k\leq N_{n}}%
}\frac{\left(  \varphi,\psi_{n,k}\right)  }{\lambda_{n}+z}\psi_{n,k}%
+\sum_{\substack{n\in\mathbb{N}_{0}\\k\leq N_{n}}}\frac{\left(  \varphi
,\psi_{n,k}\right)  }{\lambda_{n}+z}\frac{\psi_{n,k}(\underline{x}_{0}%
)}{\alpha+\frac{\sqrt{z}}{4\pi}+h(\underline{x}_{0},\underline{x}_{0},\sqrt
{z})}\mathcal{G}_{0}^{z}(\cdot,\underline{x}_{0}) \label{R_alpha_3d}%
\end{equation}
%-------------- %
for $\Omega\subset\mathbb{R}^{3}$.

\setcounter{equation}{0}
%%%%%%%%%%%%%%%%%%%%%%%%%%%%%%%%%%%%%%%%%%%%%%%%%%

\section{The principal eigenvalue of $H_{\alpha}$}

As usual, the Krein formula allows us to determine the spectrum through the
denominator of the perturbation term. In particular, it follows from the
resolvent equations (\ref{R_alpha_2d})--(\ref{R_alpha_3d}) that the spectrum
of $H_{\alpha}$ is formed by the solutions of the equations
%-------------- %
\begin{equation}
\left.
\begin{array}
[c]{l}%
\bigskip\alpha-\ln\sqrt{-\xi}-2\pi\,h(\underline{x}_{0},\underline{x}%
_{0},\sqrt{-\xi})=0\,,\qquad\Omega\subset\mathbb{R}^{2}\\
\alpha+\frac{\sqrt{-\xi}}{4\pi}+h(\underline{x}_{0},\underline{x}_{0}%
,\sqrt{-\xi})=0\,,\qquad\Omega\subset\mathbb{R}^{3}%
\end{array}
\right.  \label{eigenvalue 0}%
\end{equation}
%-------------- %
to which one has to add eigenvalues of $-\Delta_{\Omega}^{D}$, the
degenerate ones in any case and the non-degenerate ones,
$\lambda_{\bar{n}}$, provided that the corresponding eigenfunction
$\psi_{\bar{n}}$ satisfies the condition (see e.g. in
\cite{Blanchard})
%-------------- %
\begin{equation}
\psi_{\bar{n}}(\underline{x}_{0})=0
\end{equation}
%-------------- %
This does not concern, however, the bottom of the spectrum which
we are interested in here, because the ground state of
$-\Delta_{D}^{\Omega}$ is non-degenerate and can be represented by
a positive function.

The subject of this section is the principal eigenvalue of the
point-interaction operator $H_{\alpha}$. We are going to show, in particular,
that for any $\alpha\in\mathbb{R}$ there exists a unique simple eigenvalue of
$H_{\alpha}$ below the spectral threshold of $-\Delta_{\Omega}^{D}$. As a
preliminary, we need to characterize the derivatives of $h(\underline{x} _{0},
\underline{x}_{0}, \sqrt{z})$ w.r.t. the variable $\sqrt{z}$.

%-------------- %
\begin{lemma}
\label{Lemma 1} Let $z$ be a positive real number; the function $y\mapsto
h(\underline{x},\underline{x}^{\prime},y)$, defined by (\ref{h 1}) with
$y:=\sqrt{z}$, satisfies the conditions
%-------------- %
\begin{equation}
\left.
\begin{array}
[c]{l}%
\bigskip\frac{1}{y}\left(  \frac{1}{y}+2\pi\,\partial_{y}h(\underline
{x}^{\prime},\underline{x}^{\prime},y)\right)  >0\,,\qquad\Omega
\subset\mathbb{R}^{2}\\
\frac{1}{y}\left(  \frac{1}{4\pi}+\partial_{y}h(\underline{x}^{\prime
},\underline{x}^{\prime},y)\right)  >0\,,\qquad\Omega\subset\mathbb{R}^{3}%
\end{array}
\right.  \label{h 1.1}%
\end{equation}
%-------------- %
Furthermore, for any $z<0$ and $h(\underline{x}, \underline{x}^{\prime},iy)$,
defined by (\ref{h 1}) with $y:=\sqrt{\left\vert z\right\vert }$, we have%
%-------------- %
\begin{equation}
\left.
\begin{array}
[c]{l}%
\bigskip\frac{\pi}{2}+2\pi\,\operatorname{Im}h(\underline{x}^{\prime
},\underline{x}^{\prime},i\,y)=0\,,\qquad\Omega\subset\mathbb{R}^{2}\\
\frac{y}{4\pi}+\operatorname{Im}h(\underline{x}^{\prime},\underline{x}%
^{\prime},i\,y)=0\,,\qquad\Omega\subset\mathbb{R}^{3}%
\end{array}
\right.  \label{h 1.2}%
\end{equation}
%-------------- %
and
%-------------- %
\begin{equation}
\left.
\begin{array}
[c]{l}%
\bigskip\frac{1}{y}+2\pi\,\partial_{y}\operatorname{Re}h(\underline{x}%
^{\prime},\underline{x}^{\prime},i\,y)<0\,,\qquad\Omega\subset\mathbb{R}^{2}\\
\partial_{y}\operatorname{Re}h(\underline{x}^{\prime},\underline{x}^{\prime
},i\,y)<0\,,\qquad\Omega\subset\mathbb{R}^{3}%
\end{array}
\right.  \label{h 1.3}%
\end{equation}
%-------------- %
\end{lemma}
%-------------- %
\begin{proof}
We start with the 3D case. Let $z_{j}$, $j=1,2$, be a pair of positive values;
setting $y_{j}:=\sqrt{z_{j}}$ we get from equation (\ref{Green 3}) in
combination with the first resolvent formula the relation
%-------------- %
\begin{equation}
\left(  y_{1}^{2}-y_{2}^{2}\right)  \left(  \mathcal{G}_{0}^{y_{1}^{2}%
},\mathcal{G}_{0}^{y_{2}^{2}}\right)  _{L^{2}(\Omega)}=\lim_{\underline
{x}\rightarrow\underline{x}^{\prime}}\left[  \mathcal{G}^{y_{2}^{2}%
}(\underline{x},\underline{x}^{\prime})-\mathcal{G}^{y_{1}^{2}}(\underline
{x},\underline{x}^{\prime})\right]  +h(\underline{x}^{\prime},\underline
{x}^{\prime},y_{1}) -h(\underline{x}^{\prime},\underline{x}^{\prime},y_{2})\,,
\label{h 1.1_1}%
\end{equation}
%-------------- %
where $\left(  \mathcal{\cdot},\mathcal{\cdot}\right)
_{L^{2}(\Omega)}$ denotes the scalar product in $L^{2}(\Omega)$.
The limit at the r.h.s. is
easily seen to be
 %-------------- %
\begin{equation}
\lim_{\underline{x}\rightarrow\underline{x}^{\prime}}\left[  \mathcal{G}%
^{y_{2}^{2}}(\underline{x},\underline{x}^{\prime})-\mathcal{G}^{y_{1}^{2}%
}(\underline{x},\underline{x}^{\prime})\right]  =\frac{1}{4\pi}\left(
y_{1}-y_{2}\right)  \label{limite}%
\end{equation}
%-------------- %
Substituting this expression into (\ref{h 1.1_1}) we get
%-------------- %
\begin{equation}
\left(  y_{1}+y_{2}\right)  \left(  \mathcal{G}_{0}^{y_{1}^{2}},\mathcal{G}%
_{0}^{y_{2}^{2}}\right)  _{L^{2}(\Omega)}=\frac{1}{4\pi}+\frac{h(\underline
{x}^{\prime},\underline{x}^{\prime},y_{1})-h(\underline{x}^{\prime}%
,\underline{x}^{\prime},y_{2})}{y_{1}-y_{2}}\,, \label{h 1.1_2}%
\end{equation}
%-------------- %
and consequently, in the limit $y_{1}\rightarrow y_{2}$ we arrive at%
%-------------- %
\begin{equation}
\frac{1}{4\pi}+\partial_{y}h(\underline{x}^{\prime},\underline{x}^{\prime
},y_{2})=2y_{2}\,\left\Vert \mathcal{G}_{0}^{y_{2}^{2}}\right\Vert
_{L^{2}(\Omega)}^{2}>0\,.
\end{equation}
%-------------- %
On the other hand, for negative values of $z$ a similar argument yields
%-------------- %
\begin{equation}
\left(  y_{1}^{2}-y_{2}^{2}\right)  \left(  \mathcal{G}_{0}^{-y_{1}^{2}%
},\mathcal{G}_{0}^{-y_{2}^{2}}\right)  _{L^{2}(\Omega)}=\lim_{\underline
{x}\rightarrow\underline{x}^{\prime}}\left[  \mathcal{G}^{-y_{1}^{2}%
}(\underline{x},\underline{x}^{\prime})-\left(  \mathcal{G}^{-y_{2}^{2}%
}\right)  ^{\ast}(\underline{x},\underline{x}^{\prime})\right]  -h(\underline
{x}^{\prime},\underline{x}^{\prime},i\,y_{1})+h^{\ast}(\underline{x}^{\prime
},\underline{x}^{\prime},i\,y_{2}) \label{h 1.2_1}%
\end{equation}
%-------------- %
with the asterisk denoting complex conjugation and $y_{j}:=\sqrt{
\left\vert z_{j}\right\vert }$. The limit at the r.h.s. is
%-------------- %
\begin{equation}
\lim_{\underline{x}\rightarrow\underline{x}^{\prime}}\left[  \mathcal{G}%
^{-y_{1}^{2}}(\underline{x},\underline{x}^{\prime})-\left(  \mathcal{G}%
^{-y_{2}^{2}}\right)  ^{\ast}(\underline{x},\underline{x}^{\prime})\right]
=-\frac{i}{4\pi}\left(  y_{1}+y_{2}\right)  \label{limite 1}%
\end{equation}
%-------------- %
from which we get a relation replacing (\ref{h 1.1_2}), namely
%-------------- %
\begin{equation}
\left(  y_{1}^{2}-y_{2}^{2}\right)  \left(  \mathcal{G}_{0}^{-y_{1}^{2}%
},\mathcal{G}_{0}^{-y_{2}^{2}}\right)  _{L^{2}(\Omega)}=-\frac{i}{4\pi}\left(
y_{1}+y_{2}\right)  -h(\underline{x}^{\prime},\underline{x}^{\prime}%
,i\,y_{1})+h^{\ast}(\underline{x}^{\prime},\underline{x}^{\prime},i\,y_{2})\,.
\label{h 1.2_2}%
\end{equation}
%-------------- %
The second one of the relations (\ref{h 1.2}) then follows from
here by setting $y_{1}=y_{2}$. As for the real part of (\ref{h
1.2_1}) given by
%-------------- %
\begin{gather}
\frac{1}{2}\left(  y_{1}^{2}-y_{2}^{2}\right)  \left[  \left(  \mathcal{G}%
_{0}^{-y_{1}^{2}},\mathcal{G}_{0}^{-y_{2}^{2}}\right)  _{L^{2}(\Omega
)}+\left(  \mathcal{G}_{0}^{-y_{2}^{2}},\mathcal{G}_{0}^{-y_{1}^{2}}\right)
_{L^{2}(\Omega)}\right]  =\qquad\qquad\qquad\qquad\qquad\qquad\qquad
\qquad\nonumber\\
=\lim_{\underline{x}\rightarrow\underline{x}^{\prime}}\left[
\operatorname{Re}\mathcal{G}^{-y_{1}^{2}}(\underline{x},\underline{x}^{\prime
})-\operatorname{Re}\mathcal{G}^{-y_{2}^{2}}(\underline{x},\underline
{x}^{\prime})\right]  -\operatorname{Re}h(\underline{x}^{\prime},\underline
{x}^{\prime},i\,y_{1})+\operatorname{Re}h(\underline{x}^{\prime},\underline
{x}^{\prime},i\,y_{2})
\end{gather}
%-------------- %
we notice that the first term at the r.h.s. is in fact zero; dividing the
remaining ones by $y_{1}-y_{2}$ and taking the limit as $y_{1}\rightarrow
y_{2}$, we arrive at
%-------------- %
\begin{equation}
-\partial_{y}\operatorname{Re}h(\underline{x}^{\prime},\underline{x}^{\prime
},i\,y_{2})=2y_{2}\,\left\Vert \mathcal{G}_{0}^{-y_{2}^{2}}\right\Vert
_{L^{2}(\Omega)}^{2}>0\,.
\end{equation}

In the 2D case the validity of relations (\ref{h 1.1}), (\ref{h 1.2}) and
(\ref{h 1.3}) can be checked following the same idea. Taking into account the
logarithmic singularity of $\mathcal{G} ^{z}(\underline{x},\underline
{x}^{\prime})$ as $\underline{x} \rightarrow\underline{x}^{\prime}$ --- see,
e.g., \cite{Abramowitz} --- we find for $z>0$
%-------------- %
\begin{equation}
\lim_{\underline{x}\rightarrow\underline{x}^{\prime}}\left[  \mathcal{G}%
^{y_{2}^{2}}(\underline{x},\underline{x}^{\prime})-\mathcal{G}^{y_{1}^{2}%
}(\underline{x},\underline{x}^{\prime})\right]  =\frac{1}{2\pi}\left(  \ln
y_{1}-\ln y_{2}\right)  \,,
\end{equation}
%-------------- %
while for $z<0$ we have%
%-------------- %
\begin{equation}
\lim_{\underline{x}\rightarrow\underline{x}^{\prime}}\left[  \mathcal{G}%
^{-y_{1}^{2}}(\underline{x},\underline{x}^{\prime})-\left(  \mathcal{G}%
^{-y_{2}^{2}}\right)  ^{\ast}(\underline{x},\underline{x}^{\prime})\right]
=\frac{1}{2\pi}\left(  -i\pi+\ln y_{1}-\ln y_{2}\right)  \,;
\end{equation}
%-------------- %
this concludes the proof.
\end{proof}

\smallskip

%-------------- %
In the next two lemmata, we deal with the solutions of equations
(\ref{eigenvalue 0}) below the spectrum of $-\Delta_{\Omega}^{D}$. We are
going to show that for a fixed real $\alpha$ there is a unique such solution.
It is convenient to treat the 2D and 3D cases separately.
%-------------- %

\begin{lemma}
\label{Lemma 2} Let $\lambda_{0}$ denote the first eigenvalue of
$-\Delta_{\Omega}^{D}$ corresponding to the domain $\Omega\subset
\mathbb{R}^{3}$. For any $\alpha\in\mathbb{R}$, the equation
%-------------- %
\begin{equation}
\alpha+\frac{\sqrt{-\xi}}{4\pi}+h(\underline{x}_{0},\underline{x}_{0}%
,\sqrt{-\xi})=0\,,\quad\xi\in\left(  -\infty,\lambda_{0}\right)
\label{eigenvalue_3d 0}%
\end{equation}
%-------------- %
admits a unique solution, denoted $\xi(\alpha)$, such that%
%-------------- %
\begin{equation}
\lim_{\alpha\rightarrow-\infty}\xi(\alpha)=-\infty\,,\qquad\xi(-h(\underline
{x}_{0},\underline{x}_{0},0))=0\,, \label{condition 1}%
\end{equation}
%-------------- %
and
%-------------- %
\begin{equation}
\lim_{\alpha\rightarrow+\infty}\xi(\alpha)=\lambda_{0}\,. \label{condition 2}%
\end{equation}
%-------------- %
\end{lemma}

%-------------- %

\begin{proof}
In order to study solutions of (\ref{eigenvalue_3d 0}), we need to find the
dependence of $h(\underline{x}_{0},\underline{x}_{0},\sqrt{-\xi})$ on the
variable $\sqrt{-\xi}$. We start by considering the case $\xi\leq0$. Setting
$y=\sqrt{\left\vert \xi\right\vert }$, the equation (\ref{eigenvalue_3d 0})
assumes the form
%-------------- %
\begin{equation}
\frac{y}{4\pi}=-\alpha-h(\underline{x}_{0},\underline{x}_{0},y)=0\,,\quad
y\geq0 \label{eigenvalue_3d 1}%
\end{equation}
 %-------------- %
and its solutions can be geometrically interpreted as the
abscissas of the intersection points between the curves at the
left and the right-hand side of (\ref{eigenvalue_3d 1}). We will
show that for any fixed choice of $\underline{x}_{0}\in\Omega$,
$h(\underline{x}_{0},\underline{x}_{0},y)$ is a
positive and strictly decreasing function of the variable $y$, such that%
%-------------- %
\begin{equation}
\lim_{y\rightarrow+\infty}h(\underline{x}_{0}, \underline{x}_{0},y)=0\,.
\label{h_3d 1}%
\end{equation}
%-------------- %
Let us consider the boundary value problem
%-------------- %
\begin{equation}
\left\{
\begin{array}
[c]{l}%
\left(  -\Delta+y^{2}\right)  \,h(\underline{x},\underline{x}_{0},y)=0\\
\left.  h(\underline{x},\underline{x}_{0},y)\right\vert _{\underline{x}%
\in\partial\Omega}=\left.  \frac{e^{-y\left\vert \underline{x}-\underline
{x}_{0}\right\vert }}{4\pi\left\vert \underline{x}-\underline{x}%
_{0}\right\vert }\right\vert _{\underline{x}\in\partial\Omega}%
\end{array}
\right.  ;\quad\underline{x}_{0}\in\Omega\label{h_3d 2}%
\end{equation}
%-------------- %
The solution of (\ref{h_3d 2}) is infinitely smooth in the open
set $\Omega$, continuous and positive on its boundary. The strong
maximum principle --- cf.~\cite{Evans} --- in this case allows us
to claim that $h$ is strictly
positive in $\Omega$ reaching its maximum on the boundary,%
%-------------- %
\begin{equation}
0<h(\underline{x},\underline{x}_{0},y)<\sup_{\underline{x^{\prime}}\in
\partial\Omega}\frac{e^{-y\left\vert \underline{x^{\prime}}-\underline{x}%
_{0}\right\vert }}{4\pi\left\vert \underline{x^{\prime}}-\underline{x}%
_{0}\right\vert }\quad\;\mathrm{for} \;\; \forall\underline{x}\in
\Omega,\,y\geq0\,. \label{h_3d 2.1}%
\end{equation}
%-------------- %
Furthermore. the derivative $\partial_{y}h$ satisfies the equation
%-------------- %
\begin{equation}
\left\{
\begin{array}
[c]{l}%
\left(  -\Delta+y^{2}\right)  \,\partial_{y}h(\underline{x},\underline{x}%
_{0},y)=-2yh(\underline{x},\underline{x}_{0},y)\\
\left.  \partial_{y}h(\underline{x},\underline{x}_{0},y)\right\vert
_{\underline{x}\in\partial\Omega}=\left.  -\frac{e^{-y\left\vert \underline
{x}-\underline{x}_{0}\right\vert }}{4\pi}\right\vert _{\underline{x}%
\in\partial\Omega}%
\end{array}
\right.  ;\quad\underline{x}_{0}\in\Omega\label{h_3d 3}%
\end{equation}
%-------------- %
the solution of which belongs to $C^{2}(\Omega)\cap
C(\bar{\Omega})$ in view of the regularity of the source term and
the boundary value. The maximum principle --- see the version
given in \cite[Thm~IX.27]{Brezis} --- in this
case implies%
%-------------- %
\begin{equation}
\partial_{y}h(\underline{x},\underline{x}_{0},y)<0\quad\; \mathrm{for}%
\;\;\forall\underline{x} \in\Omega,\,y>0\,. \label{h_3d 4.1}%
\end{equation}
%-------------- %
In particular, the solution of (\ref{h_3d 3}) for $y=0$ is $\partial
_{y}h(\underline{x},\underline{x}_{0},0)=-\frac{1}{4\pi}$. This
characterization of $h(\underline{x}_{0},\underline{x}_{0},y)$ allows us to
claim that the equation (\ref{eigenvalue_3d 1}) admits at least one solution
for any $\alpha\in\left(  -\infty,\,-h(\underline{x}_{0},\underline{x}%
_{0},0)\right]  $ and that the conditions (\ref{condition 1}) hold. Moreover,
the second one of the relations (\ref{h 1.1}) in Lemma \ref{Lemma 1} implies
the monotonicity of the function $\alpha\mapsto y(\alpha)$ implicitly defined
by (\ref{eigenvalue_3d 1}); this grants the uniqueness of the solution.

Next we turn to (\ref{eigenvalue_3d 0}) for $\xi\in\left(  0,\lambda
_{0}\right)  $. In this case, setting $y:=\sqrt{\xi}$, the equation reads as
%-------------- %
\begin{equation}
\alpha+\frac{i\,y}{4\pi}+h(\underline{x}_{0}, \underline{x}_{0},i\,y)=0\,.
\label{eigenvalue_3d 2}%
\end{equation}
%-------------- %
According to the second one of the relations (\ref{h 1.2}), this is equivalent to%
%-------------- %
\begin{equation}
\alpha+\operatorname{Re}h(\underline{x}_{0}, \underline{x}_{0},i\,y)=0\,,
\label{eigenvalue_3d 2.1}%
\end{equation}
%-------------- %
where $\operatorname{Re}h$ satisfies the boundary value problem
%-------------- %
\begin{equation}
\left\{
\begin{array}
[c]{l}%
\left(  -\Delta-y^{2}\right)  \operatorname{Re}h(\underline{x},\underline
{x}_{0},i\,y)=0\\
\left.  \operatorname{Re}h(\underline{x},\underline{x}_{0},y)\right\vert
_{\underline{x}\in\partial\Omega}=\left.  \frac{\cos y\,\left\vert
\underline{x}-\underline{x}_{0}\right\vert }{4\pi\,\left\vert \underline
{x}-\underline{x}_{0}\right\vert }\right\vert _{\underline{x}\in\partial
\Omega}%
\end{array}
\right.  ;\qquad\underline{x}_{0}\in\Omega\label{h_3d 4.2}%
\end{equation}
%-------------- %
It is worthwhile to notice that $h(\underline{x},
\underline{x}_{0}, y)$ is not defined for $y=\sqrt{\lambda_{0}}$.
In particular, one can show that\footnote{Relation (\ref{h_3d
4.3}) easily follows, e.g., from Lemma 2 in
\cite{Blanchard}.}%
%-------------- %
\begin{equation}
\lim_{\left\vert \varepsilon\right\vert \rightarrow0}\left\Vert \varepsilon
h(\cdot,\underline{x}_{0},i\sqrt{\lambda_{0}-\varepsilon})+\psi_{0}%
(\underline{x}_{0})\,\psi_{0}(\cdot)\right\Vert _{L^{2}(U)}=0\,,
\label{h_3d 4.3}%
\end{equation}
%-------------- %
where $\psi_{0}$ is the principal eigenstate of the Dirichlet Laplacian and
$U$ is any subset of $\Omega$. In view of the boundedness of $\psi_{0}$ and
the arbitrariness of $U$, this relation also implies
%-------------- %
\begin{equation}
\lim_{y\rightarrow\sqrt{\lambda_{0}}^{-}}\left\vert h(\underline{x}%
_{0},\underline{x}_{0},i\,y)\right\vert =+\infty\label{h_3d 4.4}%
\end{equation}
%-------------- %
Using this result together with the conditions (\ref{h 1.2}) and
(\ref{h 1.3})
of Lemma \ref{Lemma 1}, we conclude that $\operatorname{Re}h(\underline{x}%
_{0},\underline{x}_{0},i\,y)$ is a strictly decreasing function of
$y\in\left(  0,\sqrt{\lambda_{0}}\right)  $ whose behavior for $y\rightarrow
\sqrt{\lambda_{0}}$ is given by%
%-------------- %
\begin{equation}
\lim_{y\rightarrow\sqrt{\lambda_{0}}^{-}}\operatorname{Re}h(\underline{x}%
_{0},\underline{x}_{0},i\,y)=-\infty\label{h_3d 4.5}%
\end{equation}
%-------------- %
Summing up this discussion, the equation (\ref{eigenvalue_3d 2.1}) has a
unique positive solution $y=y(\alpha)$ for any $\alpha\in\left(
-h(\underline{x} _{0},\underline{x}_{0}, 0),\,+\infty\right)  $, and this
solution asymptotically approaches the value $\lambda_{0}$ as $\alpha
\rightarrow+\infty$; in combination with the first part this concludes the
proof of the lemma.
\end{proof}

\smallskip

Next we deal with the eigenvalue equation in the two-dimensional case. Recall
that the free Green's function related to this problem is the modified Bessel
function $K_{0}$, which is strictly positive and convex in $\mathbb{R}^{+}$
and admits the following representation \cite{Abramowitz}
%-------------- %
\begin{equation}
K_{0}(x)=-\left\{  \ln\frac{x}{2}+\gamma\right\}  I_{0}(x)+\sum_{n=1}%
^{+\infty}c_{n}\,\frac{x^{2n}}{\left(  2\,n!\right)  ^{2n}}\,, \label{K_0}%
\end{equation}
%-------------- %
where $\gamma\approx0.577$ is the Euler-Mascheroni constant,
$c_{n}=\sum _{k=1}^{n}\frac{1}{k}$, and $I_{0}(x)$ is the other
modified Bessel function
given by the series%
%-------------- %
\begin{equation}
I_{0}(x)=\sum_{n=0}^{+\infty}\,\frac{x^{2n}}{\left(  2n!\right)  ^{2n}}\,.
\label{I_0}%
\end{equation}
%-------------- %
In the following proof we will make use of the asymptotic properties of
$K_{0}$,
%-------------- %
\begin{equation}
\lim_{x\rightarrow0^{+}}K_{0}(x)=+\infty\,, \quad\lim_{x\rightarrow+\infty}
K_{0}(x)=0\,. \label{K_0 1}%
\end{equation}

%-------------- %
\begin{lemma}
Let $\lambda_{0}$ denote the first eigenvalue of $-\Delta_{\Omega}^{D}$
corresponding to the domain $\Omega\subset\mathbb{R}^{2}$. For any $\alpha
\in\mathbb{R}$, the equation
%-------------- %
\begin{equation}
\bigskip\alpha-\ln\sqrt{-\xi}-2\pi\,h(\underline{x}_{0},\underline{x}
_{0},\sqrt{-\xi})=0\,,\quad\xi\in\left(  -\infty,\lambda_{0} \right)  \,,
\label{eigenvalue_2d 0}%
\end{equation}
%-------------- %
admits an unique solution, denoted $\xi(\alpha)$, such that%
%-------------- %
\begin{equation}
\lim_{\alpha\rightarrow+\infty}\xi(\alpha)=-\infty\,, \qquad\xi(f(\underline
{x}_{0},\underline{x}_{0},0))=0\,, \label{condition 1.1}%
\end{equation}
%-------------- %
and
%-------------- %
\begin{equation}
\lim_{\alpha\rightarrow-\infty}\xi(\alpha)=\lambda_{0}\,,
\label{condition 2.1}%
\end{equation}
%-------------- %
where $f(\underline{x},\underline{x}_{0},\sqrt{-\xi})$ in the
second one of the relations (\ref{condition 1.1}) is a
$C^{+\infty}(\Omega)\cap C(\bar{\Omega})$--regular func\-tion of
both the spatial variables defined by
%-------------- %
\begin{equation}
f(\underline{x},\underline{x}_{0},\sqrt{-\xi})=2\pi\,h(\underline
{x},\underline{x}_{0},\sqrt{-\xi})+\ln\sqrt{-\xi}\,I_{0}(\sqrt{-\xi
}\,\left\vert \underline{x}-\underline{x}_{0}\right\vert )\,,\quad\xi
<\lambda_{0}\,. \label{f_2d}%
\end{equation}
%-------------- %
\end{lemma}

%-------------- %

\begin{proof}
The argument follows the same line as in Lemma \ref{Lemma 2}, the main
difference coming from the specific form of the Green functions in two
dimensions. For $\xi\leq0$ and $y:=\sqrt{\left\vert \xi\right\vert }$,
equation (\ref{eigenvalue_2d 0}) reads
%-------------- %
\begin{equation}
\bigskip\alpha=\ln y+2\pi\,h(\underline{x}_{0},\underline{x}_{0},y)\,,
\label{eigenvalue_2d 1}%
\end{equation}
%-------------- %
where $h(\underline{x},\underline{x}_{0},\sqrt{y})$, the solution
of the boundary value problem
%-------------- %
\begin{equation}
\left\{
\begin{array}
[c]{l}%
\left(  -\Delta+y^{2}\right)  \,h(\underline{x},\underline{x}_{0},y)=0\\
\left.  h(\underline{x},\underline{x}_{0},y)\right\vert _{\underline{x}%
\in\partial\Omega}=\left.  \frac{1}{2\pi}K_{0}(y\,\left\vert \underline
{x}-\underline{x}_{0}\right\vert )\right\vert _{\underline{x}\in\partial
\Omega}%
\end{array}
\right.  ;\quad\underline{x}_{0}\in\Omega\label{h_2d 1}%
\end{equation}
%-------------- %
is strictly positive in $\Omega$, as it follows from the maximum principle and
the positivity of the boundary values. Consequently, the r.h.s. of
(\ref{eigenvalue_2d 1}) diverges as $y\rightarrow+\infty$. Moreover, from
(\ref{h 1.1}) we know that $\ln y+2\pi\,h(\underline{x}_{0},\underline{x}%
_{0},y)$ is strictly increasing as a function of $y$ in the whole
$\mathbb{R}^{+}$. In order to study the behaviour for $y\rightarrow0$, we
introduce the auxiliary function $f(\underline{x},\underline{x}_{0},y)$,
defined in (\ref{f_2d}) which solves the boundary value problem
%-------------- %
\begin{equation}
\left\{
\begin{array}
[c]{l}%
\left(  -\Delta+y^{2}\right)  \,f(\underline{x},\underline{x}_{0},y)=0\\
\left.  h(\underline{x},\underline{x}_{0},y)\right\vert _{\underline{x}%
\in\partial\Omega}=\left.  \frac{1}{2\pi}K_{0}(y\,\left\vert \underline
{x}-\underline{x}_{0}\right\vert )+\ln y\,I_{0}(\sqrt{-\xi}\,\left\vert
\underline{x}-\underline{x}_{0}\right\vert )\right\vert _{\underline{x}%
\in\partial\Omega}%
\end{array}
\right.  ;\quad\underline{x}_{0}\in\Omega\label{f_2d 1}%
\end{equation}
%-------------- %
Using the definition (\ref{I_0}), it is easy to verify that the
quantity $\ln y+2\pi\,h(\underline{x}_{0},\underline{x}_{0},y)$
coincides with the value of
$f$ in $\underline{x}_{0}$, thus for $y\rightarrow0^{+}$ we have%
%-------------- %
\[
\lim_{y\rightarrow0^{+}}\left(  \ln y+2\pi\,h(\underline{x}_{0},\underline
{x}_{0},y)\right)  =\lim_{y^{+}\rightarrow0}f(\underline{x}_{0},\underline
{x}_{0},y)\,.
\]
%-------------- %
In this limit, the problem (\ref{f_2d 1}) simplifies to
%-------------- %
\begin{equation}
\left\{
\begin{array}
[c]{l}%
-\Delta\,f(\underline{x},\underline{x}_{0},0)=0\\
\left.  h(\underline{x},\underline{x}_{0},0)\right\vert _{\underline{x}%
\in\partial\Omega}=\left.  -\frac{1}{2\pi}\left\{  \ln\frac{\left\vert
\underline{x}-\underline{x}_{0}\right\vert }{2}+\gamma\right\}  \right\vert
_{\underline{x}\in\partial\Omega}%
\end{array}
\right.  ;\quad\underline{x}_{0}\in\Omega\label{f_2d 1.1}%
\end{equation}
 %-------------- %
where the representation (\ref{K_0}) has been considered; the
regularity of the boundary condition in (\ref{f_2d 1.1}) implies
$f(\underline{x} ,\underline{x}_{0},0)\in C^{+\infty}(\Omega)\cap
C(\bar{\Omega})$. So far we have shown that $\left(  \ln
y+2\pi\,h(\underline{x}_{0},\underline{x} _{0},y)\right)  $ is a
strictly increasing function of $y\in\mathbb{R}^{+}$
such that%
 %-------------- %
\begin{equation}
\lim_{y\rightarrow+\infty}\left(  \ln y+2\pi\,h(\underline{x}_{0}%
,\underline{x}_{0},y)\right)  =+\infty\label{condition 1.1.1}%
\end{equation}
%-------------- %
and
%-------------- %
\begin{equation}
\lim_{y\rightarrow0^{+}}\left(  \ln y+2\pi\,h(\underline{x}_{0},\underline
{x}_{0},y)\right)  =f(\underline{x}_{0},\underline{x}_{0},0)\,.
\label{condition 1.1.2}%
\end{equation}
 %-------------- %
From here we can conclude that the equation (\ref{eigenvalue_2d
0}) admits an unique solution, $\xi(\alpha)\leq0$, for any
$\alpha\in\left[  f(\underline{x} _{0},\underline{x}_{0},
0),+\infty\right)  $, which satisfies the conditions
(\ref{condition 1.1}).

In the case of a positive eigenvalue $\xi\in\left(  0,\lambda_{0}\right)  $,
setting $y=\sqrt{\xi}$ we rewrite equation (\ref{eigenvalue_2d 0}) as
%-------------- %
\begin{equation}
\alpha=\ln i\,y+2\pi\,h(\underline{x}_{0},\underline{x}_{0},i\,y)
\label{eigenvalue_2d 2}%
\end{equation}
 %-------------- %
which due to (\ref{h 1.2}) is equivalent to%

\begin{equation}
\alpha=\ln y+2\pi\operatorname{Re}h(\underline{x}_{0},\underline{x}%
_{0},i\,y)\,. \label{eigenvalue_2d 2.1}%
\end{equation}
%-------------- %
As in the 3D case, we notice that the r.h.s of (\ref{eigenvalue_2d 2.1}) is a
strictly decreasing function, cf. (\ref{h 1.3}), diverging as $y\rightarrow
\sqrt{\lambda_{0}}$; the sought conclusion easily follows.
\end{proof}

\smallskip

Since the spectrum of $H_{\alpha}$ is determined by the solutions of the
equations (\ref{eigenvalue 0}), the above lemmata have the following
implication which means that in a sense point interactions in dimension two
and three can be always regarded as ``attractive''.
%-------------- %

\begin{corollary}
For any real $\alpha$, the operator $H_{\alpha}$ defined by
(\ref{H_alpha_2d dom})--(\ref{H_alpha}) has a unique spectral point below the
spectral threshold of $-\Delta_{\Omega}^{D}$.
\end{corollary}

\setcounter{equation}{0}
%%%%%%%%%%%%%%%%%%%%%%%%%%%%%%%%%%%%%%%%%%%%%%%%%%

\section{Dependence of the principal eigenvalue on the position of the
interaction}

Now we pass to our main topic. We will characterize the behaviour of the
principal eigenvalue of the point-interaction Hamiltonians $H_{\alpha}$ for a
fixed bounded domain $\Omega\subset\mathbb{R}^{d}$, $d=2,3$, as the
interaction site moves towards the boundary of $\Omega$. We will restrict our
attention to domains having an interior reflection property w.r.t. a suitable
hyperplane, in the following sense:

\begin{definition}
Consider a hyperplane $P$ of dimension $d-1$ in $\mathbb{R}^{d}$ and denote by
$S^{P}$ the mirror image of a set $S\subset\mathbb{R}^{d}$ w.r.t. $P$ provided
$S\cap P=\varnothing$. The domain $\Omega$ is said to have the \emph{interior
reflection} property w.r.t. $P$ if $P\cap\Omega\neq\varnothing$ and there is
an open connected component $\Omega_{s}\subset\Omega\backslash P$ such that
$\Omega_{s}^{P}$ is a proper subset of $\Omega\backslash\bar{\Omega}_{s}$. We
call $\Omega_{s}$ the \emph{smaller side} of $\Omega$ and $P$ an
\emph{interior reflection} hyperplane.
\end{definition}

To prove our main result, we need following auxiliary statement.

\begin{lemma}
\label{Lemma 3}Let $\mathcal{G}_{0}^{z}(\underline{x},\underline{x}^{\prime})$
be defined by (\ref{Green 3})-(\ref{h 1}) and $z\in\mathbb{R}$. For values of
$z$ above $-\lambda_{0}$, the following implications hold,
%-------------- %
\begin{equation}
z\in\left[  0,+\infty\right)  \: \Longrightarrow\:\mathcal{G}_{0}%
^{z}(\underline{x},\underline{x}^{\prime})>0\quad\; \mathrm{in}\ \Omega
\label{Green 1.1}%
\end{equation}
 %-------------- %
and%
 %-------------- %
\begin{equation}
z\in\left(  -\lambda_{0},0\right)  :\Longrightarrow:\operatorname{Re}%
\mathcal{G}_{0}^{z}(\underline{x},\underline{x}^{\prime})>0\quad
\;\mathrm{in}\ \Omega\,.\label{Green 1.2}%
\end{equation}
 %-------------- %
\end{lemma}

\begin{proof}
For $\Omega\subset\mathbb{R}^{3}$ and $z\in\left[  0,+\infty\right)  $, the
Green function%
 %-------------- %
\begin{equation}
\mathcal{G}_{0}^{z}(\underline{x},\underline{x}^{\prime})=\frac{e^{-\sqrt
{z}\left\vert \underline{x}-\underline{x}^{\prime}\right\vert }}%
{4\pi\left\vert \underline{x}-\underline{x}^{\prime}\right\vert }%
-h(\underline{x},\underline{x}^{\prime},\sqrt{z})\label{Green_3d 3.1}%
\end{equation}
 %-------------- %
is certainly positive in a small enough open neighbourhood
$B_{\underline{x} ^{\prime}}$ of the point
$\underline{x}^{\prime}$ due to the boundedness of
$h(\underline{x}, \underline{x}^{\prime},\sqrt{z})$. Moreover, it
solves the
boundary value problem%
 %-------------- %
\begin{equation}
\left\{
\begin{array}
[c]{l}%
\left(  -\Delta+z\right)  \mathcal{G}_{0}^{z}(\underline{x},\underline
{x}^{\prime})=0\quad in\ \Omega\backslash B_{\underline{x}^{\prime}}\\
\left.  \mathcal{G}_{0}^{z}\right\vert _{\partial\Omega}=0\,, \quad\left.
\mathcal{G}_{0}\right\vert _{\partial B_{\underline{x}^{\prime}}}>0
\end{array}
\right.  ;\qquad\underline{x}^{\prime}\in\Omega\label{Green_3d 4}%
\end{equation}
%-------------- %
It follows from the maximum principle that $\mathcal{G}_{0}^{z}$ is strictly
positive in the whole $\Omega$.

To prove the other implication in the 3D case, notice that for $z\in\left(
-\lambda_{0},0\right)  $ we have
%-------------- %
\begin{equation}
\operatorname{Re}\mathcal{G}_{0}^{z}(\underline{x},\underline{x}^{\prime
})=\frac{\cos\sqrt{\left\vert z\right\vert }\left\vert \underline
{x}-\underline{x}^{\prime}\right\vert }{4\pi\left\vert \underline
{x}-\underline{x}^{\prime}\right\vert }-\operatorname{Re}h(\underline
{x},\underline{x}^{\prime},i\sqrt{\left\vert z\right\vert })\,.
\label{Green_3d 3.2}%
\end{equation}
 %-------------- %
Once more we can find a suitable open neighbourhood of the point
$\underline{x}^{\prime}$, which we call
$B_{\underline{x}^{\prime}}$, where this function is positive. In
$\Omega\backslash B_{\underline{x}^{\prime}}$
$\operatorname{Re}\mathcal{G}_{0}^{z}(\underline{x},\underline{x}^{\prime})$
solves the boundary value problem%

\begin{equation}
\left\{
\begin{array}
[c]{l}%
\left(  -\Delta+z\right)  \operatorname{Re}\mathcal{G}_{0}^{z}(\underline
{x},\underline{x}^{\prime})=0\quad in\ \Omega\backslash B_{\underline
{x}^{\prime}}\\
\left.  \mathcal{G}_{0}\right\vert _{\partial\Omega}=0\,, \quad\left.
\mathcal{G}_{0}\right\vert _{\partial B_{\underline{x}^{\prime}}}>0
\end{array}
\right.  ;\qquad\underline{x}^{\prime}\in\Omega\label{Green_3d 5}%
\end{equation}
%-------------- %
Under the condition $z>-\lambda_{0}$ we can still apply the maximum principle
obtaining in this way $\operatorname{Re} \mathcal{G}_{0} ^{z}(\underline
{x},\underline{x}^{\prime})>0$ in$\ \Omega$. Finally, in the 2D case the proof
follows the same line with the replacement (\ref{Green_3d 3.1}) and
(\ref{Green_3d 3.2}) by the corresponding 2D Green's function
%-------------- %
\begin{equation}
\mathcal{G}_{0}^{z}(\underline{x},\underline{x}^{\prime})=K_{0}(\sqrt
{z}\left\vert \underline{x}-\underline{x}^{\prime}\right\vert )-h(\underline
{x},\underline{x}^{\prime},\sqrt{z}) \label{Green_2d 3.1}%
\end{equation}
 %-------------- %
and taking into account the asymptotic properties of $K_{0}(x)$ as
$x\rightarrow0$.
\end{proof}

\smallskip

Now we are in position to prove our main result. The next theorem shows that,
under the interior reflection conditions imposed on the domain $\Omega$, the
principal eigenvalue of $H_{\alpha}$ increases as the interaction site moves
towards the boundary of the smaller side $\Omega_{s}$ of $\Omega$.

\begin{theorem}
\label{th: main} \label{Theorem 1}Let $P$ be an interior reflection hyperplane
for the domain $\Omega$ and denote by $\underline{n}$ the normal vector to $P$
pointing towards $\Omega_{s}$. Assume that $\underline{x}_{0}\in\Omega
\cap\left(  \partial\Omega_{s}\cap P\right)  $; then the principal eigenvalue
$\xi$ of the point-interaction $H_{\alpha}$ with the perturbation placed at
$\underline{x}_{0}$ satisfies the condition
%-------------- %
\begin{equation}
\underline{n}\cdot\nabla_{\underline{x}_{0}}\xi>0\,. \label{eigenvalue_3d 3.1}%
\end{equation}
 %-------------- %
\end{theorem}

%-------------- %

\begin{proof}
Consider first the 3D case. To analyze the dependence of $\xi$ on the
interaction position $\underline{x}_{0}$, we have to distinguish between the
negative and positive spectral points. If $\alpha\in\left(  -\infty
,\,-h(\underline{x}_{0},\underline{x}_{0},0)\right]  $, then $\xi
=\xi(\underline{x}_{0})$ is by Lemma \ref{Lemma 2} a negative solution of
(\ref{eigenvalue_3d 1}). Replacing $\sqrt{\left\vert \xi\right\vert }$ with
$y(\underline{x}_{0})$ and taking the gradient w.r.t. $\underline{x}_{0}$ in
(\ref{eigenvalue_3d 1}) we find
%-------------- %
\begin{equation}
\nabla_{\underline{x}_{0}}y(\underline{x}_{0})\left(  \frac{1}{4\pi}%
+\partial_{y}h\right)  =-\nabla_{\underline{x}_{0}}h(\underline{x}%
_{0},\underline{x}_{0},y)\,. \label{eigenvalue_3d 3.3}%
\end{equation}
 %-------------- %
Next we consider the term $\nabla_{\underline{x}_{0}}
h(\underline{x} _{0},\underline{x}_{0},y)$ at the r.h.s. of the
last equation; under our interior reflection assumptions we will
show that this vector is oriented towards the smaller side of
$\Omega$. To this aim we notice that, in view of the relations
(\ref{Green 2})--(\ref{Green free_3d}), $h(\underline{x}
,\underline{x}^{\prime},y)$ can be written as%
 %-------------- %
\begin{equation}
h(\underline{x},\underline{x}^{\prime},y)=\frac{e^{-y\left\vert \underline
{x}-\underline{x}^{\prime}\right\vert }}{4\pi\left\vert \underline
{x}-\underline{x}^{\prime}\right\vert }-\sum_{\substack{n\in\mathbb{N}%
_{0}\\k\leq N_{n}}}\frac{\psi_{n,k}(\underline{x}^{\prime})\,\psi
_{n,k}(\underline{x})}{\lambda_{n}+y^{2}} \label{h_3d 5.1}%
\end{equation}
%-------------- %
for any $\underline{x}\neq\underline{x}^{\prime}$. From the symmetry of this
expression and the regularity $h(\underline{x},\underline{x}^{\prime},y)$ it
follows that
%-------------- %
\begin{equation}
\left.  \nabla_{\underline{x}^{\prime}}h(\underline{x},\underline{x}^{\prime
},y)\right\vert _{\underline{x}=\underline{x}^{\prime}=\underline{x}_{0}%
}=\left.  \nabla_{\underline{x}}h(\underline{x},\underline{x}^{\prime
},y)\right\vert _{\underline{x}=\underline{x}^{\prime}=\underline{x}_{0}}
\label{h_3d 6}%
\end{equation}
 %-------------- %
and%
 %-------------- %
\begin{equation}
\nabla_{\underline{x}_{0}}h(\underline{x}_{0},\underline{x}_{0},y)=2\left.
\nabla_{\underline{x}}h(\underline{x},\underline{x}^{\prime},y)\right\vert
_{\underline{x}=\underline{x}^{\prime}=\underline{x}_{0}}\,. \label{h_3d 7}%
\end{equation}
%-------------- %
To analyze the orientation of this vector, we introduce the function $u$
defined on the smaller part of $\Omega$ by
%-------------- %
\begin{equation}
u(\underline{x},\underline{x}_{0},y)=h(\underline{x},\underline{x}%
_{0},y)-h(\underline{x}^{P},\underline{x}_{0},y),\quad\underline{x}\in
\Omega_{s}\,, \label{u_3d 0}%
\end{equation}
 %-------------- %
where $\underline{x}^{P}$ denotes the mirror image of
$\underline{x} \in
\Omega_{s}$ through the plane $P$. The following equation holds%
 %-------------- %
\begin{equation}
\left\{
\begin{array}
[c]{l}%
\left(  -\Delta+y^{2}\right)  u=0\qquad\text{in }\Omega_{s}\\
\left.  u\right\vert _{P\cap\Omega}=0\,,\quad\left.  u\right\vert
_{\partial\Omega_{s}\cap\partial\Omega}=\left.  \frac{e^{-y\left\vert
\underline{x}-\underline{x}_{0}\right\vert }}{4\pi\left\vert \underline
{x}-\underline{x}_{0}\right\vert }-h(\underline{x}^{P},\underline{x}%
_{0},y)\right\vert _{\underline{x}\in\partial\Omega s\cap\partial\Omega}%
\end{array}
\right.  ; \quad\underline{x}_{0}\in\Omega\cap P \label{u_3d}%
\end{equation}
%-------------- %
It is worthwhile to notice that the boundary value on $\partial\Omega_{s}%
\cap\partial\Omega$ can be identified with the value of $\mathcal{G}%
_{0}^{y^{2}}(\underline{x},\underline{x}_{0})$ on the set $\left(
\partial\Omega s\cap\partial\Omega\right)  ^{P}$, indeed we have
%-------------- %
\begin{equation}
\left.  \mathcal{G}_{0}^{y^{2}}(\underline{x},\underline{x}_{0})\right\vert
_{\left(  \partial\Omega s\cap\partial\Omega\right)  ^{P}}=\left.
\frac{e^{-y\left\vert \underline{x}-\underline{x}_{0}\right\vert }}%
{4\pi\left\vert \underline{x}-\underline{x}_{0}\right\vert }-h(\underline
{x},\underline{x}_{0},y)\right\vert _{\left(  \partial\Omega s\cap
\partial\Omega\right)  ^{P}}=\left.  \frac{e^{-y\left\vert \underline
{x}-\underline{x}_{0}\right\vert }}{4\pi\left\vert \underline{x}-\underline
{x}_{0}\right\vert }-h(\underline{x}^{P},\underline{x}_{0},y)\right\vert
_{\partial\Omega s\cap\partial\Omega} \label{condition 3}%
\end{equation}
 %-------------- %
Then it follows from (\ref{Green 1.1}) that $u$ is positive on
$\partial\Omega
s\cap\partial\Omega$ and by the maximum principle, $u>0$ holds in $\Omega_{s}%
$. In particular, $u$ reaches its minimum on the points of the open surface
$P\cap\Omega$; the Hopf boundary-point lemma in this case implies%
 %-------------- %
\begin{equation}
\underline{n}\cdot\nabla_{\underline{x}}u>0\quad\; \mathrm{for}\;\;\forall
\underline{x}\in P\cap\Omega\,. \label{u_3d 1}%
\end{equation}
%-------------- %
Due to the definition (\ref{u_3d 0}), in combination with the relation
%-------------- %
\begin{equation}
\underline{n}\cdot\nabla_{\underline{x}}h(\underline{x}^{P},\underline{x}%
_{0},y)=-\underline{n}\cdot\nabla_{\underline{x}^{P}}h(\underline{x}%
^{P},\underline{x}_{0},y)\quad in\ \Omega_{s}\cup\left(  P\cap\Omega\right)
\,,
\end{equation}
 %-------------- %
the last inequality also implies%
 %-------------- %
\begin{equation}
2\,\underline{n}\cdot\nabla_{\underline{x}}h(\underline{x},\underline{x}%
_{0},y)>0\quad\;\mathrm{for}\;\;\forall\underline{x}\in P\cap\Omega
\,.\label{h_3d 7.1}%
\end{equation}
%-------------- %
Substituting (\ref{h_3d 7}) and (\ref{h_3d 7.1}) into the r.h.s. of
(\ref{eigenvalue_3d 3.3}) and taking the projection in the direction of the
vector $\underline{n}$ we get
%-------------- %
\begin{equation}
\underline{n}\cdot\nabla_{\underline{x}_{0}}y(\underline{x}_{0})\left(
\frac{1}{4\pi}+\partial_{y}h\right)  =-\underline{n}\cdot\nabla_{\underline
{x}_{0}}h(\underline{x}_{0},\underline{x}_{0},y)<0\,.\label{eigenvalue_3d
3.4}%
\end{equation}
 %-------------- %
The term $\nabla_{\underline{x}_{0}}y(\underline{x}_{0})$ at the
l.h.s. of
(\ref{eigenvalue_3d 3.4}) is related to $\nabla_{\underline{x}_{0}}\xi$ by%
 %-------------- %
\begin{equation}
\nabla_{\underline{x}_{0}}y=-\frac{1}{y}\nabla_{\underline{x}_{0}}%
\xi\label{eigenvalue_3d 3.4.1}%
\end{equation}
%-------------- %
from which it follows that
%-------------- %
\begin{equation}
\frac{1}{y}\left(  \frac{1}{4\pi}+\partial_{y}h\right)  \nabla_{\underline
{x}_{0}}\xi>0 \label{eigenvalue_3d 3.5}%
\end{equation}
 %-------------- %
The sought inequality (\ref{eigenvalue_3d 3.1}) follows easily
from (\ref{eigenvalue_3d 3.5}) taking into account the condition
(\ref{h 1.1}).

In the opposite case, $\alpha> -h(\underline{x}_{0}, \underline{x}_{0},0)$,
the first spectral point of $H_{\alpha}$ is a strictly positive solution of
the equation
%-------------- %
\begin{equation}
\alpha+\operatorname{Re}h(\underline{x}_{0},\underline{x}_{0},i\,\sqrt{\xi})=0
\label{eigenvalue_3d 4.1}%
\end{equation}
 %-------------- %
with $\xi<\lambda_{0}$ --- cf.~(\ref{eigenvalue_3d 2.1}) in Lemma
\ref{Lemma 2}. Replacing $\sqrt{\xi}$ with $y(\underline{x}_{0})$
and taking
the gradient w.r.t. $\underline{x}_{0}$ in (\ref{eigenvalue_3d 4.1}), we get%
 %-------------- %
\begin{equation}
\nabla_{\underline{x}_{0}}y\,\partial_{y} \operatorname{Re}h(\underline{x}
_{0},\underline{x}_{0},i\,y)= -\nabla_{\underline{x}_{0}}\operatorname{Re}
h(\underline{x}_{0},\underline{x}_{0},i\,y)\,. \label{eigenvalue_3d 4.2}%
\end{equation}
%-------------- %
In order to check the orientation of the vector at the r.h.s of this
expression, we notice again that due to the symmetry of the function
%-------------- %
\begin{equation}
\operatorname{Re}h(\underline{x},\underline{x}^{\prime},i\,y)=\frac{\cos
y\,\left\vert \underline{x}-\underline{x}^{\prime}\right\vert }{4\pi\left\vert
\underline{x}-\underline{x}^{\prime}\right\vert }-\sum_{\substack{n\in
\mathbb{N}_{0}\\k\leq N_{n}}}\frac{\psi_{n,k}(\underline{x}^{\prime}%
)\,\psi_{n,k}(\underline{x})}{\lambda_{n}-y^{2}} \label{h_3d 5.2}%
\end{equation}
 %-------------- %
the gradient $\nabla_{\underline{x}_{0}}\operatorname{Re}h(\underline{x}%
_{0},\underline{x}_{0},i\,y)$ can be expressed as%
 %-------------- %
\begin{equation}
\nabla_{\underline{x}_{0}}\operatorname{Re}h(\underline{x}_{0},\underline
{x}_{0},i\,y)=2\left.  \nabla_{\underline{x}}\operatorname{Re}h(\underline
{x},\underline{x}^{\prime},y)\right\vert _{\underline{x}=\underline{x}
^{\prime}=\underline{x}_{0}}\,. \label{h_3d 7.2}%
\end{equation}
%-------------- %
Then we follow the line of the first part of the proof introducing the
function $u$,
%-------------- %
\begin{equation}
u(\underline{x},\underline{x}_{0},y)=\operatorname{Re}h(\underline
{x},\underline{x}_{0},i\,y)-\operatorname{Re}h(\underline{x}^{P},\underline
{x}_{0},i\,y)\,,\quad\underline{x}\in\Omega_{s}\,. \label{u_3d 2}%
\end{equation}
 %-------------- %
Proceeding as before and taking into account the implication
(\ref{Green 1.2}) we find easily
 %-------------- %
\begin{equation}
\underline{n}\cdot\nabla_{\underline{x}}u>0
\end{equation}
%-------------- %
and
%-------------- %
\begin{equation}
2\,\underline{n}\cdot\nabla_{\underline{x}}\operatorname{Re}h(\underline
{x},\underline{x}_{0},y)>0\quad\; \mathrm{for}\;\; \forall\underline{x}\in
P\cap\Omega\,; \label{h_3d 7.3}%
\end{equation}
 %-------------- %
substituting (\ref{h_3d 7.2}) and (\ref{h_3d 7.3}) into the r.h.s.
of (\ref{eigenvalue_3d 4.2}), we conclude that
 %-------------- %
\begin{equation}
\underline{n}\cdot\nabla_{\underline{x}_{0}}y\,\partial_{y}\operatorname{Re}%
h(\underline{x}_{0},\underline{x}_{0},i\,y)<0 \label{eigenvalue_3d 4.3}%
\end{equation}
%-------------- %
The claim (\ref{eigenvalue_3d 3.1}) is then obtained from
(\ref{eigenvalue_3d 4.3}) by taking into account the relation
%-------------- %
\begin{equation}
\nabla_{\underline{x}_{0}}y=\frac{1}{y}\nabla_{\underline{x}_{0}}\xi
\end{equation}
 %-------------- %
and the inequality (\ref{h 1.3}) from Lemma \ref{Lemma 1}. This
concludes the argument in the three-dimensional case; the
two-dimensional one can be dealt with in the same way, step by
step.
\end{proof}

\setcounter{equation}{0}
%%%%%%%%%%%%%%%%%%%%%%%%%%%%%%%%%%%%%%%%%%%%%%%%%%

\section{Optimization of $\xi(\underline{x}_{0})$}

By Theorem \ref{Theorem 1} the spectral threshold of the operator $H_{\alpha}$
increases as the interaction position $\underline{x}_{0}$ moves towards the
boundary of the domain $\Omega$. This result provides us with some insights on
how to place the point-interaction centre to minimize the principal eigenvalue
of the Hamiltonian $H_{\alpha}$. For the sake of simplicity, we begin with the
case of a convex $\Omega$. Let $\Pi$ be the set of all the hyperplanes $P$ of
interior reflection for $\Omega$; we denote by $\Omega_{s,P}$ the smaller part
related to $P\in\Pi$, provided it exists, and by $\Sigma$ the union%
\begin{equation}
\Sigma={\bigcup\limits_{P\in\Pi}}\Omega_{s,P} \label{Sigma}%
\end{equation}

The following claim is a straightforward consequence of the Theorem
\ref{Theorem 1}.

\begin{corollary}
\label{Lemma 4}Let $\Omega$ be an open convex domain in $\mathbb{R}^{d} $,
$d=2,3$, and let $H_{\alpha}$ be a point-interaction operator in $\Omega$ with
the perturbation placed at $\underline{x}_{0}$. The principal eigenvalue of
$H_{\alpha}$, considered as a function of the interaction centre, takes its
minimum value when $\underline{x}_{0}$ belongs to the open set $\Omega
\backslash\Sigma$.
\end{corollary}

%-------------- %

\begin{proof}
Notice first that the continuity of $h(\underline{x},\underline{x}_{0},y)$
implies the continuity of solutions of the eigenvalue equations
(\ref{eigenvalue 0}), thus the principal eigenvalue $\xi(\underline{x}_{0})$
has at least one minimum point $\underline{x}_{0}^{m}\in\bar{\Omega}$. We use
reduction ad absurdum: assume $\underline{x}_{0}^{m}\in\Omega_{s,P}$ for a
suitable hyperplane $P\in\Pi$. Due to the convexity of the domain, it exists
another hyperplane $P^{\prime}\in\Pi$ parallel to $P$ and such that
 %-------------- %
\begin{equation}
\underline{x}_{0}^{m}\in\partial\Omega_{s,P^{\prime}}\cap P^{\prime}\,,
\label{condition 4}%
\end{equation}
 %-------------- %
however, under this assumption Theorem \ref{Theorem 1} implies the
inequality
 %-------------- %
\begin{equation}
\underline{n}\cdot\nabla_{\underline{x}_{0}}\xi>0 \label{condition 4.1}%
\end{equation}
 %-------------- %
from which a contradiction follows easily.
\end{proof}

In the case of highly symmetric domains such as the interior of a circle or an
ellipse in the plane, and similarly a ball and an interior of an ellipsoid in
three dimensions, it is easy to identify the set $\Sigma$ with the center of
such a domain. More generally, the convexity of $\Omega$ ensures the validity
of the interior reflection property needed in Theorem~\ref{Theorem 1} with
respect to some hyperplane passing through a point sufficiently close to the
boundary. In this situation the above result can be used to localize the
optimal position of $\underline{x}_{0}$ in a `central' subset of the domain
$\Omega$.

A slight generalization of the above argument lead us to an analogous rule to
localize the minimum points of the principal eigenvalue for a point
interaction within non-convex domains. Let $P\in\Pi$ and consider the one
parameter family of hyperplanes $P_{t}$,
%-------------- %
\begin{equation}
\left\{
\begin{array}
[c]{l}%
\bigskip P_{t}=\left\{  \underline{x}+\underline{n}\,t:\underline{x}\in
P,\ t\in\left[  0,T\right]  \right\} \\
T=\max\left\{  t\in\mathbb{R}^{+}:P_{t}\cap\Omega\neq\varnothing\right\}
\end{array}
\right.  \label{P_t def}%
\end{equation}
 %-------------- %
where $\underline{n}$ denotes the unit normal to $P$ directed
towards the smaller part $\Omega_{s}$ of $\Omega$. We denote as
$\Pi^{\prime}$ the subset
formed by all those hyperplanes $P\in\Pi$ such that%

\begin{equation}
P_{t}\in\Pi\quad\;\mathrm{for}\;\;\forall\,t\in\left[  0,T\right]  \,.
\label{condition 5}%
\end{equation}
%-------------- %
It is important to notice that in the non-convex case, to any hyperplane of
interior reflection there may correspond more than one smaller part. Next we
denote by $\Theta_{s,P}$ the union of all the smaller parts related to $P$,
and by $\Sigma^{\prime}$ the set
%-------------- %
\begin{equation}
\Sigma^{\prime}={%
%TCIMACRO{\dbigcup \limits_{P\in\Pi^{\prime}}}%
%BeginExpansion
{\displaystyle\bigcup\limits_{P\in\Pi^{\prime}}}
%EndExpansion
}\Theta_{s,P}\,. \label{Sigma'}%
\end{equation}
 %-------------- %

\begin{corollary}
\label{Lemma 5}Assume that $\Omega$ is an open domain in $\mathbb{R}^{d}$,
$d=2,3$, and $H_{\alpha}$ is a point-interaction operator in $\Omega$ with the
perturbation placed at $\underline{x}_{0}$. The principal eigenvalue of
$H_{\alpha}$, regarded as a function of $\underline{x}_{0}$, takes its minimum
value when $\underline{x}_{0}$ belongs to the open set $\Omega\backslash
\Sigma^{\prime}$.
\end{corollary}
%-------------- %

\begin{proof}
The argument is an easy modification of the proof of Corollary~\ref{Lemma 4}.
\end{proof}

\smallskip

Non-convex domains with a reasonably simple boundary such as, for instance,
the union of two intersecting disks or a dog-bone profile in two dimensions,
can be easily analyzed using Corollary~\ref{Lemma 5}. It is also worthwhile to
stress that the results of this section do not depend on the parameter
$\alpha$, hence the optimal placement of the point interaction with respect to
the minimum of the principal eigenvalues can be the same irrespective of the
interaction ``strength''.

Let us finally comment on he relation to the work \cite{Harrell} mentioned in
the introduction. We have said that for a hard-wall obstacle the principal
eigenvalue \emph{decreases} as it moves towards the boundary. The difference
of the two effects can be traced back to the different boundary conditions
which characterize the operator domains in the two cases. While the hard
obstacle is characterized by Dirichlet boundary condition, the
point-interaction operator $H_{\alpha}$ considered here can be obtained as the
norm-resolvent limit of a family of sphere interactions Hamiltonians
$H_{\alpha}(r)$ with the boundary condition of a \emph{mixed type} as the
radius $r\rightarrow0$. In the three-dimensional case, for instance, the
operator $H_{\alpha}(r)$ is explicitly given by
%-------------- %
\[
\left\{
\begin{array}
[c]{l}%
H_{\alpha}(r)=-\Delta\qquad\mathrm{on}\;\,\Omega\backslash S_{r}\\
D(H_{\alpha}(r))=\left\{  \left.  \psi\in H^{1}(\mathbb{R}^{3})\right\vert
\,\left(  \partial_{n}\psi\right)  _{+}-\left(  \partial_{n}\psi\right)
_{-}=\frac{1}{4\pi\alpha\,r^{2}+r}\psi\right\}
\end{array}
\right.
\]
%-------------- %
where $S_{r}$ denotes the sphere of radius $r$ centered at
$\underline{x}_{0}$ and $\left(  \partial_{n}\psi\right)
_{+}-\left(  \partial_{n}\psi\right) _{-}$ is the jump of the
normal derivative of $\psi$ on the interaction surface
\cite{Figari 1, Shimada}. Another insight into the difference of
the two situations can be obtained from \cite{ES96}.

%%%%%%%%%%%%%%%%%%%%%%%%%%%%%%%%%%%%%%%%%%%

\subsection*{Acknowledgments}

We wish to thank Rodolfo Figari and Francesco Chiacchio for their
useful remarks. This research was partially supported by GAAS and
MEYS of the Czech Republic under projects A100480501 and LC06002.


\bigskip

\begin{thebibliography}{9999999999}                                                                                       %


\bibitem[AS72]{Abramowitz}M. Abramowitz and I.A. Stegun, ed.: \emph{Handbook
of Mathematical functions}, Dover, New York 1972.
%-------------- %


\bibitem[AG63]{Akhiezer}N.I. Akhiezer, I.M. Glazman. \emph{Theory of Linear
Operators in Hilbert Space, Vol II}, Frederick Ungar Publishing Co., New York
1963.
%-------------- %


\bibitem[AGHH05]{Albeverio}S. Albeverio, F. Gesztesy, R H\"{o}gh-Krohn and H.
Holden: \emph{Solvable Models in Quantum Mechanics, 2nd edition, with an
appendix by P. Exner}, AMS, Providence, R.I, 2005.
%-------------- %


\bibitem[AB92a]{AB92a}M.S.~Ashbaugh, R.D.~Benguria: A sharp bound for the
ratio of the first two eigenvalues of Dirichlet Laplacians and extensions,
\emph{Ann. Math.} \textbf{135} (1992), 601-628.
%-------------- %


\bibitem[AB92b]{AB92b}M.S.~Ashbaugh, R.D.~Benguria: A second proof of the
Payne-P\'olya-Weinberger conjecture, \emph{Commun.Math.Phys.} \textbf{147}
(1992), 181-190.
%-------------- %


\bibitem[BFM07]{Blanchard}Ph. Blanchard, R. Figari, A. Mantile: Point
interaction Hamiltonians in bounded domains, \emph{J. Math. Phys.} \textbf{48}
(2007), 082108
%-------------- %


\bibitem[Br83]{Brezis}H. Brezis. \emph{Analyse fonctionelle,} Masson, Paris
1983.
%-------------- %


\bibitem[Ev98]{Evans}L.C. Evans. \emph{Partial Differential Equations}, AMS,
Providence, R.I., 1998.
%-------------- %


\bibitem[EG\v{S}T96]{EGST96}P.~Exner, R.~Gawlista, P.~\v Seba, M.~Tater: Point
interactions in a strip, \emph{Ann. Phys.} \textbf{252} (1996), 133-179.
%-------------- %


\bibitem[EN02]{EN02}P.~Exner, K.~N\v{e}mcov\'{a}: Quantum mechanics of layers
with a finite number of point perturbations, \emph{J.Math.Phys.} \textbf{43}
(2002), 1152-1184.
%-------------- %


\bibitem[E\v{S}96]{ES96}P.~Exner, P.~\v Seba: Point interactions in dimension
two and three as models of small scatterers, \emph{Phys. Lett.} \textbf{A222}
(1996), 1-4.
%-------------- %


\bibitem[Fa23]{Fa23}G.~Faber: Beweiss das unter allen homogenen Membranen von
Gleicher Fl\"ache und gleicher Spannung die kreisf\"ormige den Tiefsten
Grundton gibt, \emph{Sitzungber. der math.-phys. Klasse der Bayerische Akad.
der Wiss. zu M\"unchen} (1923), 169-172.
%-------------- %


\bibitem[FT93]{Figari 1}R. Figari, A. Teta: A boundary value problem of mixed
type on perforated domains, \emph{Asymptotic Analysis} \textbf{6} (1993),
271-284.
%-------------- %


\bibitem[HKK01]{Harrell}E.M. Harrell, P. Kr\"{o}ger, K. Kurata: On the
placement of an obstacle or a well so as to optimize the fundamental
eigenvalue, \emph{SIAM J. Math. Anal} \textbf{33} (2001), 240-259.
%-------------- %


\bibitem[Kr25]{Kr25}E.~Krahn: \"Uber eine von Rayleigh formulierte minimal
Eigenschaft des Kreises, \emph{Ann. Math.} \textbf{94} (1925), 97-100.
%-------------- %


\bibitem[PPW55]{PPW55}L.E.~Payne, G.~P\'olya, H.F.~Weinberger: Sur le quotient
de deux fr\'equences propres consecutives, \emph{Comp. Rend.} \textbf{241}
(1955), 917-919.
%-------------- %


\bibitem[RS75]{Reed2}M. Reed, B. Simon. \emph{Methods of Modern Mathematical
Physics, II. Fourier Analysis, Self-Adjointness,} Academic Press, New York
1975.
%-------------- %


\bibitem[RS78]{Reed4}M. Reed, B. Simon. \emph{Methods of Modern Mathematical
Physics, IV. Analysis of Operators,} Academic Press, New York 1978.
%-------------- %


\bibitem[Sh03]{Shimada}Shin-ichi Shimada: Resolvent convergence of sphere
interactions to point interactions, \emph{J. Math. Phys.} \textbf{44} (2003), 990-1005.
\end{thebibliography}
\end{document}